\documentclass[final]{IEEEtran}
\usepackage{etex}
\usepackage{caption}
\usepackage{relsize}
\usepackage{amsthm,amssymb,graphicx,graphicx,amsmath,color}
\usepackage{tikz,pgfplots,epstopdf}
\usetikzlibrary{arrows}
\usetikzlibrary{shapes.geometric,shapes.arrows,decorations.pathmorphing}
\usetikzlibrary{matrix,chains,scopes,positioning,arrows,fit}
\usetikzlibrary{shapes.geometric, arrows} 
 \usepackage[usenames,dvipsnames]{pstricks}
 \usepackage{epsfig}
 \usepackage{pst-grad} 
\usepackage{pst-plot} 

\usepackage{color}

\def\E{\mathbb{E}}

\def\re{\mathtt{Re}}

\def\ie{{\em i.e.}}

\def\Z{\mathbb{Z}}

\def\sinc{\operatorname{sinc}}

\def\d{\mathrm{d}}

\def\l{\ell}

\newtheorem{lemma}{Lemma}{}

 \def\LPF{\operatorname{LPF}}
\def\conv{\otimes}
\title{A Linearization Technique for Self-Interference Cancellation in Full-Duplex Radios }
\author{\IEEEauthorblockN{Arjun Nadh,  Joseph Samuel, Ankit Sharma, S. Aniruddhan,  Radha Krishna Ganti}\\
\IEEEauthorblockA{Department of Electrical Engineering\\
Indian Institute of Technology, Madras\\
Chennai, India 600036\\
\{arjunnadh, jsamuel, ankitsharma, ani, rganti\}@ee.iitm.ac.in}
\thanks{A part of this work has been presented in NCC 2016. \cite{ncc16}. }
}

\def\antenna{%
    -- +(0mm,4.0mm) -- +(2.625mm,7.5mm) -- +(-2.625mm,7.5mm) -- +(0mm,4.0mm)
}

\tikzstyle{box1} = [rectangle, minimum width=3cm, minimum height=1cm, text centered, draw=black, fill=orange!30]
\tikzstyle{box2} = [rectangle, minimum width=1cm, minimum height=1cm, text centered, draw=black, fill=blue!30]
\tikzstyle{circ1} = [circle,   radius=3cm,  text centered, draw=black, fill=orange!30]
\tikzstyle{circ2} = [circle,   radius=3cm,  text centered, draw=black, fill=blue!30]
\tikzstyle{triag1} = [ regular polygon, regular polygon sides=3, inner sep=1.5pt,draw = black, fill=blue!30]

\tikzstyle{output} = [coordinate]
\tikzstyle{arrowdouble} = [thick,->,>=stealth,double]
\tikzstyle{arrow} = [thick,->,>=stealth]
\tikzstyle{null} = [ text centered]

\begin{document}

\maketitle
\begin{abstract}
The fundamental problem in the design of a full-duplex radio is the cancellation of the self-interference (SI) signal generated by the transmitter.
 Current techniques for suppressing SI rely on generating a copy of the SI signal and subtracting it partly in the RF (radio frequency) and digital domains.   A critical step in replicating the  self-interference is the estimation of the  multi-path channel  through which the transmitted signal propagates to the antenna. Since there is no prior model on the number of multipath reflections, current techniques assume a tap delay line filter (in the RF and digital domain) with a large number of taps, and estimate the taps in the analog and the digital domain. 

In this paper,  using a linearization technique, we show that the self-interference channel in an indoor environment  can be effectively modelled  as $H(f)=C_0 + C_1f$ in the frequency domain. Thus, the effective self-interference channel can be represented by two parameters $C_0$ and $C_1$, irrespective of the multipath environment.  We also provide experimental evidence to verify the above channel model and propose novel low-complexity designs for self-interference cancellation.

  \end{abstract}
\section{Introduction}
 Full-duplex wireless communication involves transmitting and receiving at the same time and at the same frequency. An ideal full-duplex communication system has double the usable bandwidth  in a  bi-directional link. However, self-interference is a major impediment in realizing a full-duplex system as the much stronger transmit signal drowns the received signal and in the process saturates the receiver chain.   In the past few years, there has been renewed efforts  in  building an ideal full-duplex system \cite{FD:5G_beyond, asu:14, kim_2015}.  
\subsection{Origin of self-interference}
The signal coming into the receiver of a full duplex system will have not just the intended received signal but also many copies of the transmitted signal coming through various paths. These paths arise because of the following reasons:
\begin{enumerate}
\item Transmitter and receiver coupling at the antenna. 
\label{coupling_at_antenna}
\item Leakage of the transmit signal from the power amplifier to the receiver. This can happen via the printed circuit board (PCB) substrate (for a system implemented on a board) or the silicon substrate (for a system implemented on an integrated circuit).
\label{leakage_thru_substrate}
\item Reflections of the transmitted signal from the external environment picked up by the antenna.  These reflections are typically attenuated due to path-loss and therefore the strength of the reflected signal depends on the distance to  the reflector.
\label{wireless_reflections}
\end{enumerate}
While the self-interference signal components arising out of \ref{coupling_at_antenna}) and \ref{leakage_thru_substrate}) can be estimated a priori through calibration, the characteristics of the self-interference paths caused \ref{wireless_reflections}) are random and depend on the geometry of the reflectors. 

In current approaches, the self-interference is cancelled at multiple stages beginning with the RF stage followed by cancellation in the baseband analog and digital domains. We will now briefly review some of the existing self-interference cancellation techniques.
\subsection{RF cancellation techniques}
In a two antenna system, antenna separation is a simple way of providing passive isolation between the transmit and the receive chains \cite{fd:sab_14, fd:par_07}. The isolation depends on the separation distance between antenna, orientation and polarization. In general, the isolation can be upwards of 40 dB.  However, in this paper we focus on shared antenna architecture wherein a single antenna is used for both transmission and reception.

In  \cite{FD:rice,FD:duarte2010full}  the entire transmit chain is replicated for generating a duplicate copy  of the self-interference signal for cancellation.  However, the additional RF chain introduces  noise and canceling the non-linearities introduced by the PA in the transmit path is much more difficult.    A common technique for RF self-interference cancellation in shared antenna architectures is to use a multitap RF filter with fixed delay lines and tunable gains. In \cite{FD:bharadia2013full}, sixteen RF delay  lines with variable gains were used to filter the known RF signal.  The delays (in the range of 400 ps to 1.4 ns) were permanently tuned to the strongest self-interference paths through the PCB and the antenna.  An RF cancellation of about 60 dB is reported (in conjunction with a circulator).  In \cite{DBLP:journals/corr/HuusariCLKTV15},  a three-path  RF filter using vector modulators is implemented. Also, the control logic of the vector modulators is implemented in the analog domain and an RF cancellation of  60 dB for 20 MHz is reported.  This technique of implementing an RF filter (tap-delay line filter) has the following disadvantages:
  \begin{enumerate}
 \item It requires multiple delay lines, and achieving delays at high frequencies in the RF domain  is extremely difficult. 
 \item It is difficult to realize such a multi-tap RF filter in a small form factor. 
 \item The RF lines have to be carefully tuned a priori and the resulting circuit might not work effectively if the reflectors and the self-interference paths change substantially. 
 \end{enumerate}

 A 35-40 dB passive isolation was reported in \cite{FD:kumar2014directional, FD:knox2012single} using hybrid transformers and a vector modulator.  An electrical balance tunable  RF network was used in \cite{FB:van2014rf} to provide an isolation of greater than 50 dB in the RF domain with a combination of active and passive techniques. 

In summary, the passive isolation techniques (separation of antenna, circulator) can provide about  40 dB isolation \cite{FD:knox2012single}, the active cancellation techniques provide about 35-40 dB isolation providing the total RF isolation to be about 65-70 dB for 20 MHz bandwidth.

\subsection{Baseband analog cancellation}
In \cite{FD:kaufman2013analog}, an analog cancellation technique is proposed in conjunction with an antenna design. However, the gains of the analog cancellation in isolation  are not clear. Analog cancellation techniques while recognized as important, are not common because of the restricted  access to the baseband signals in commercial off-the-shelf (COTS) equipment. 

\subsection{Baseband digital cancellation}
Digital baseband cancellation consists of removing the residual self-interference after RF and baseband analog cancellation.
In \cite{FD:bharadia2013full, FD:duarte2010full}, the SI channel  is estimated \cite{FD:bharadia2013full} using a least-squares technique and the SI is cancelled using the estimated channel and the known transmitted signal. However, these techniques incur significant complexity since the entire channel (with unknown number of taps) has to be estimated constantly  to track the channel changes due to the varying reflections. The importance of removing the non-linear components of the signal are highlighted in \cite{FD:bharadia2013full} and  $45$ dB digital cancellation was reported. Other implementations \cite{FD:duarte2010full} have  reported about $30$ dB  cancellation.  In \cite{FD:day2012full}, it has been shown that the limited dynamic range of the analog-to-digital conversion is a bottleneck in effective cancellation of self-interference in the digital domain. In \cite{FD:choi2013simultaneous}, the system level performance of full-duplex radios is presented. In all these digital techniques, no prior model of the filter (for the linear components) is used leading to a higher implementation complexity. Digital cancellation leads to about 35-40 dB of self-interference suppression for most of these designs.

\subsection{Our contribution}
All   current implementations of self-interference cancellation utilize some variable  delay elements to match the known transmit signal with the interference signal. These delay elements might be variable length delay-lines in the RF domain or instead utilize  advanced baseband signal processing to compensate for the delay. The current implementations assume no prior model of the self-interference channel and hence incur large complexity in terms of channel estimation either in the RF domain (circuitry) or in the digital domain  (tune a large number of variables). 
\begin{itemize}
\item In this paper, we introduce a new technique of self-interference cancellation based on linearization of the delayed signal.  In the simplest case, the linearization technique leads to a two-parameter channel model that leads to a simpler implementation both in analog (RF) and digital domains.   
\item No prior model of the channel is assumed in terms of number of taps or the dominant paths. Hence, the  proposed method is robust to changes in the environment, \ie,  it can adapt to varying reflector profiles (RF multipath).
\item  It is experimentally verified that the  proposed technique leads to almost 75 dB of self-interference cancellation for 20 MHz signal (without cancellation of non-linear components). 
\end{itemize}
 
In Section \ref{sec:intro}, a generic model for self-interference based on the transmitted signal in full-duplex nodes is provided. In Section \ref{sec:Sec_BasicIdea}, the basic idea of linearization  of self-interference based on Taylor series approximation is introduced. In Section \ref{sec:model}, analog and digital self-interference cancellation architectures based on the linearization technique are presented. The experimental results are presented in Section \ref{sec:exp} and the paper is concluded in Section \ref{sec:con}.

\section{Modeling of self-interference}
\label{sec:intro}
 We now begin with the description of the self-interference problem in full-duplex systems assuming an ideal baseband to RF conversion in the transmit chain and an ideal downconversion in the receive chain. The transmitted RF signal $y(t)$ equals
 \[y(t) = \sqrt{G_t}\re[x(t) e^{j 2\pi f_c t}],\]
 where $x(t)$ is the complex baseband signal generated by the digital-to-analog converter and $f_c$ is the RF transmission frequency. In practice, the bandwidth  $W$ of the baseband signal is much less than the transmit frequency $f_c$. For example, in an 802.11a system $W=5$ MHz and $f_c = 2.5$ GHz. Hence, we assume $f_c \gg W$.   We also assume a direct conversion receiver, which is the most prevalent receiver architecture in current wireless systems.
 Similarly,  we assume an ideal  RF chain in the reciever which implies that the downconversion is equivalent to 
 \[ r(t)= \LPF(z(t)e^{-j 2 \pi f_ct}),\]
 where $z(t)$ is the signal at the input to the low noise amplifier (LNA), and  $r(t)$ is the received baseband signal. In the above equation $\LPF ()$ is the low pass filtering operation in the downconversion chain.  The baseband  signal $x(t)$, in a single-carrier system,  is  modelled as 
\begin{align}
x(t)=   \sum_{n=-\infty}^{\infty} s_n g \left(\frac{t-nT}{T}\right),
\label{eq:bb1}
\end{align}
where $s_n$ denote the information symbols, $g(t)$ is the pulse shaping filter and $T$ is the symbol duration (also known as signaling time). Commonly used pulse shaping filters are root-raised-cosine and sinc pulses.  A similar model for the baseband signal in an OFDM system is provided in  \cite{chiueh2012baseband}. We also have  $T\approx  W^{-1}$, and in this paper,  we assume an equality\footnote{In particular, based on the modulation,   an appropriate scaling constant  has to be introduced. However, this is of not much relevance in the current paper. }.  We also assume  $\E|x(t)|^2 =1$. 
 

%
%
 

 The self-interference seen at the  antenna will be a sum of delayed and attenuated copies of $y(t)$. 
More precisely, the self-interference at the input to the mixer (before downconversion) can be   modelled as 
\begin{align}
I(t) &= \sum_{k=1}^M a_k y(t-\tau_k),\nonumber\\
&=\sum_{k=1}^M a_k \sqrt{G_t}\re[ x(t-\tau_k) e^{j 2\pi f_c( t-\tau_k)}].
\label{eq:1}
\end{align}
where $a_k$ represent the attenuation of the path $k$ and $\tau_k$ denotes the delay in the path $k$.  

  In the model \eqref{eq:1} for $I(t)$,  for simplicity,  we  consider only $M$ significant paths or reflections. Observe that $M$, $a_i$ and $\tau_i$, $1\leq i\leq M$, are unknown and depend on the geometry of the reflectors.  In reality, there might be a large number of reflections. However, as mentioned before, most of these reflections might be insignificant due path loss.  Without loss of generality, we will assume that $a_1\geq a_2\geq \hdots \geq a_M$. 

If $c$ is the speed of the light, the round trip distance for a path $k$ is $c\tau_k$. Assuming a path-loss function $\l(x)$, we have  $a_k = \sqrt{\l\left( c\tau_k \right)}$.
\begin{figure}[ht]
\begin{center}
\includegraphics[width=0.95\columnwidth]{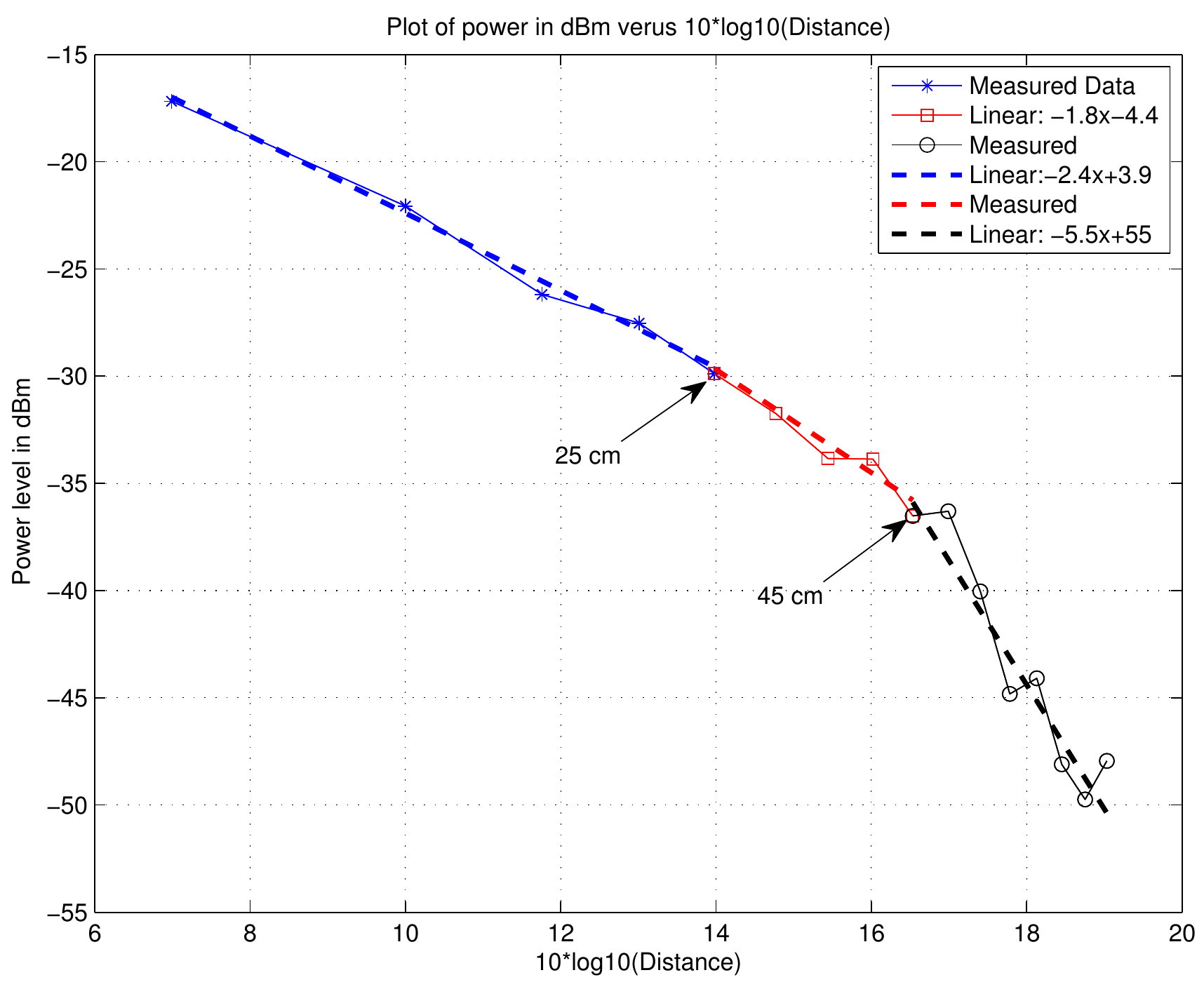}
\end{center}
\caption{Path loss as a function of distance (in cm) obtained by measurements at 2.5 GHz with an isotropic antenna.}
\label{fig:1}
\end{figure}
A common path loss function is $\l(x) =\min\{\delta, k\|x\|^{-\alpha}\}$, where $\alpha >2$, and $\delta$ and $k$ are constants with appropriate dimensions.  A measured path loss model  at 2.5 GHz is shown in Figure \ref{fig:1}. We observe that for a reflector at 12.5 cm, the path loss is about 30 dB. The time delay $\tau$ for this distance is about 830 ps. 
In general, the values of $\tau_k$ which correspond to significant powers in the self-intereference signal are small compared to the signalling time, \ie,  $\tau_k \ll T$. This is because the paths corresponding to smaller values of $\tau_k$ will have smaller path losses.  

The self-interference $I(t)$ can be rewritten as
\begin{align}
I(t) = y(t)\conv h_r(t),
\end{align}
where $h_r(t) = \sum_{k=1}^M a_k \delta(t-\tau_k)$ is an RF filter. In indoor environments,  looking at the values of $\tau_k$ that contribute significantly to the self-interference, we can observe that the filter $h_r(t)$ has a bandwidth of few tens of gigahertz. This is in sharp contrast to the standard channel estimation problem in digital communications where the filter bandwidth is of the order of the bandwidth of the signal. 
At the node, the transmitted RF signal $y(t)$  and the baseband signal $x(t)$ are known. However, the channel $h_r(t)$ is unknown making the self-interference cancellation non-trivial. 

  \section{Basic idea: Linearization of the self-interference signal}
\label{sec:Sec_BasicIdea}
We now introduce the technique of linearization of the self-interference signal that would form the basis of our new cancellation technique. 
\subsection{First order approximation}
  We use a first order Taylor series approximation for the delayed signal $x(t-\tau_k)$ to obtain 
\begin{align}
x(t- \tau_k) = x(t) -x'(t) \tau_k +E(\tau_k,t),
\label{tay1}
\end{align}
where $E(\tau_k,t)$ is the residual error, and $x'(t)=\frac{\d x(t)}{\d t}$ is the derivative of the baseband signal. The error term decreases with decreasing $\tau_k$ improving the approximation. We will now express the self-interference signal as a function of the baseband signal $x(t)$ and its derivatives. From \eqref{eq:1}, we have 
\begin{align*}
I(t)&=\sqrt{G_t}\sum_{k=1}^M a_k \re[ x(t-\tau_k) e^{j 2\pi f_c( t-\tau_k)}]\\
    &\stackrel{(a)}{=} \sqrt{G_t}\sum_{k=1}^M a_k \re[ ( x(t) - x'(t) \tau_k +E(\tau_k,t)) e^{j 2\pi f_c( t-\tau_k)}]\\
&=\sqrt{G_t} \re \Big[\sum_{k=1}^M a_k x(t) e^{j 2\pi f_c t} e^{-j 2\pi f_c \tau_k} \\
&-  \sum_{k=1}^M a_k x'(t) \tau_k e^{j 2\pi f_c t} e^{-j 2\pi f_c \tau_k} \Big] +E_2(t)\\
  &=\sqrt{G_t} \re\Big [ x(t)e^{j2\pi f_c t}C_0 - x'(t)e^{j2\pi f_c t}C_1 \Big]+E_2(t),
\end{align*}	
where \[C_0= \sum_{k=1}^Ma_k e^{-j2\pi f_c\tau_k},\] and \[C_1= \sum_{k=1}^Ma_k\tau_k e^{-j2\pi f_c\tau_k}.\] Here  $(a)$ follows from \eqref{tay1} and $E_2(t)$ is the error term given by
 \begin{align}
E_2(t)=\sqrt{G_t} \sum_{k=1}^M a_k \re[ E(\tau_k,t) e^{j 2\pi f_c( t-\tau_k)}].
\end{align}
In the above expression, observe that $C_0$ and $C_1$ are complex numbers that are constant (for a given geometry of backscatters)   and depend only  on the path-loss and the path delays. 
Hence the self-interference can be  expressed as 
\begin{align}
I(t) = I_s(t)-I_{d}(t)+ E_2(t),
\label{eqn:Eqn_SelfInterf}
\end{align}
where $I_s(t)$ $=$  $\re[C_0  x(t)e^{j2\pi f_c t} ]$,  represents the self-interference because of
the signal and $I_d(t)$ $=$ $\re[C_1 x'(t) e^{j2\pi f_c t}]$
represents the self-interference because of the derivative term. If the error term   $E_2(t)$ is small, canceling $I(t)$ would amount to canceling the interference because of $I_s(t)$  and $I_d(t)$ which only depend on the signal $x(t)$ and its derivative $x'(t)$.
In other words, we do not require to construct many different delayed copies of $x(t)$ to cancel the self-interference signal -- we just need the signal $x(t)$ and its derivative $x'(t)$. The signal $x(t)$ is readily available in the analog domain since that is the signal being transmitted. The derivative $x'(t)$ can easily be generated in the analog or digital domain (which will be explained in the next Section).  Neglecting $E_2(t)$, the channel can now be modelled as $H(f)=C_0 + C_1 f$. The effect of the entire channel is now captured by just two complex numbers $C_0$ and $C_1$, thus providing a huge dimensionality reduction -- the number of parameters to estimate is only two irrespective of the nature and number of multipaths.

The Taylor series is a good approximation for paths arising with low $\tau_k$ and the error is large for paths with large $\tau_k$. However, paths with large $\tau_k$ arise from far-off  reflectors and hence the path loss for these paths is large, and therefore the absolute error contribution  by these paths is low. Hence the overall approximation of $I_t$  by $I_s$ and $I_d$ is reasonable for signals with low bandwidth. We will make this observation precise in the next Lemma by bounding the power of the error.
\begin{lemma} 
\label{lem:der}
For the single-carrier baseband transmission  with sinc pulses and $\E|s_k|^2$=1 (average energy of the signal constellation is unity),  the average power of the  error term in the Taylor series is given by
\[\E[|a_kE(\tau_k, t)|^2]  \leq  (0.075)a_k ^2(  \tau_k/T)^4 .\]
\end{lemma}
\begin{proof}
By Taylor's  theorem, there exists a $\delta_k \in (t-\tau_k,t)$ such that
\[E(\tau_k,t)  = \frac{1}{2}   x''(\delta_k) \tau_k^2.\]
Hence \begin{align}
\E[|a_kE(\tau_k,t)|^2 ] &= \frac{1}{4}  a_k^2  \tau_k^4 \E|x''(\delta_k)|^2.
\label{eq:avg}
\end{align} 
Using \eqref{eq:bb1}, we have
\[x''(\delta_k) = \sum_{n=-\infty}^\infty s_n \frac{\d^2 g((\delta_k-nT)/T)}{\d t}.\]
Assuming $\E[s_i]=0$, $\E[|s_i|^2]=1$ and that $s_i$'s are uncorrelated, we have
\begin{align*}
\E|x''(\delta_k)|^2 &=   \sum_{n=-\infty}^\infty \left(\frac{\d^2 g((\delta_k-nT)/T)}{\d t}\right)^2\\
&= \frac{1}{T^4}  \sum_{n=-\infty}^\infty f(\delta_k/T -n),
\end{align*}
where  $(a)$ follows by twice differentiating the  sinc pulse  $g(t)$ $= \sinc(t)$ and 
\[f(x)= \left(\frac{2  \sin(x)}{x^3}-\frac{\sin( x)}{x}-\frac{2  \cos(x)}{x ^2}\right)^2.\] 
%
%
   It can be easily verified that $f(x)\leq\frac{1}{x^2}$ and hence $f(x)$ is absolutely summable.  We therefore have
   \begin{align*}
   \sum_{n=-\infty}^\infty f(x+n)
   &\stackrel{(a)}{=}  \sum_{n=-\infty}^\infty\hat{f}(n)e^{2\pi j n x},
   \end{align*}
   where $(a)$ follows from the Poisson summation formula \cite{folland2013real} and 
   \[\hat{f}(\xi)= \int_{-\infty}^\infty  f(x) e^{-j 2\pi x \xi} \d x,\]
   is the Fourier transform of the function $f(x)$.  Also,
      \begin{align*}
   \hat{f}(0)&= \frac{1}{5}\sqrt{\frac{\pi}{2}},\\
   \hat{f}(-1)=  \hat{f}(1)&=\frac{1}{60}\sqrt{\frac{\pi}{2}},\\
         \hat{f}(k)&=0, |k|\geq 2, k \in \Z.
   \end{align*}
   Hence,
      \begin{align*}
      \sum_{n=-\infty}^\infty f(\delta_k/T -n) &= \frac{1}{5}\sqrt{\frac{\pi}{2}}+ \frac{2}{60}\sqrt{\frac{\pi}{2}}\cos(2\pi \delta_k/T),\\
      &\leq  \frac{1}{5}\sqrt{\frac{\pi}{2}}+ \frac{2}{60}\sqrt{\frac{\pi}{2}} \leq 0.3.
      \end{align*}
   Plugging into \eqref{eq:avg} we obtain the result.
\end{proof}
From the above lemma we observe that the error term depends on the ratio of $\tau_k /T$ and its product with the channel gain $a_k^2$. Hence the total error power due to the Taylor series approximation can be bounded as
\[\E |E_2(t)^2| \leq  (0.075)\sum_{k=0}^M a_k ^2(  \tau_k/T)^4 . \]
Hence we see that the Taylor series approximation is good as long as the path loss $a_k^2$ and the reflector distances ($c/(2\tau_k)$) in relation to the bandwidth $1/T$ are such that $a_k ^2(  \tau_k/T)^4$ is small. 
More insight can be obtained by converting all the terms into distances. The delay $\tau_k$ is given by $d_k/c$, where $d_k/2$ is the distance to the reflector. Assuming a path-loss model $\l(x) =C\|x\|^{-\alpha}$, $\alpha>2$ we have 
\[\E[|a_kE(\tau_k,t)|^2] =\Theta(d_k^{-\alpha+4} ( cT)^{-4} ).\]
So if $\alpha\geq 4$,  the contribution of the reflector decreases with distance irrespective of $T$. For $2\leq \alpha <4$, the error is low as long as $Cd_k^{4-\alpha} \ll (cT)^{-4}$. For a 20 MHz signal $(cT)^{-4}$ is about -47 dB.  Also, we can observe that the error term increases with decreasing signalling time $T$ or equivalently increasing bandwidth.

\begin{figure}[ht]
\centering
\includegraphics[width=0.95\columnwidth]{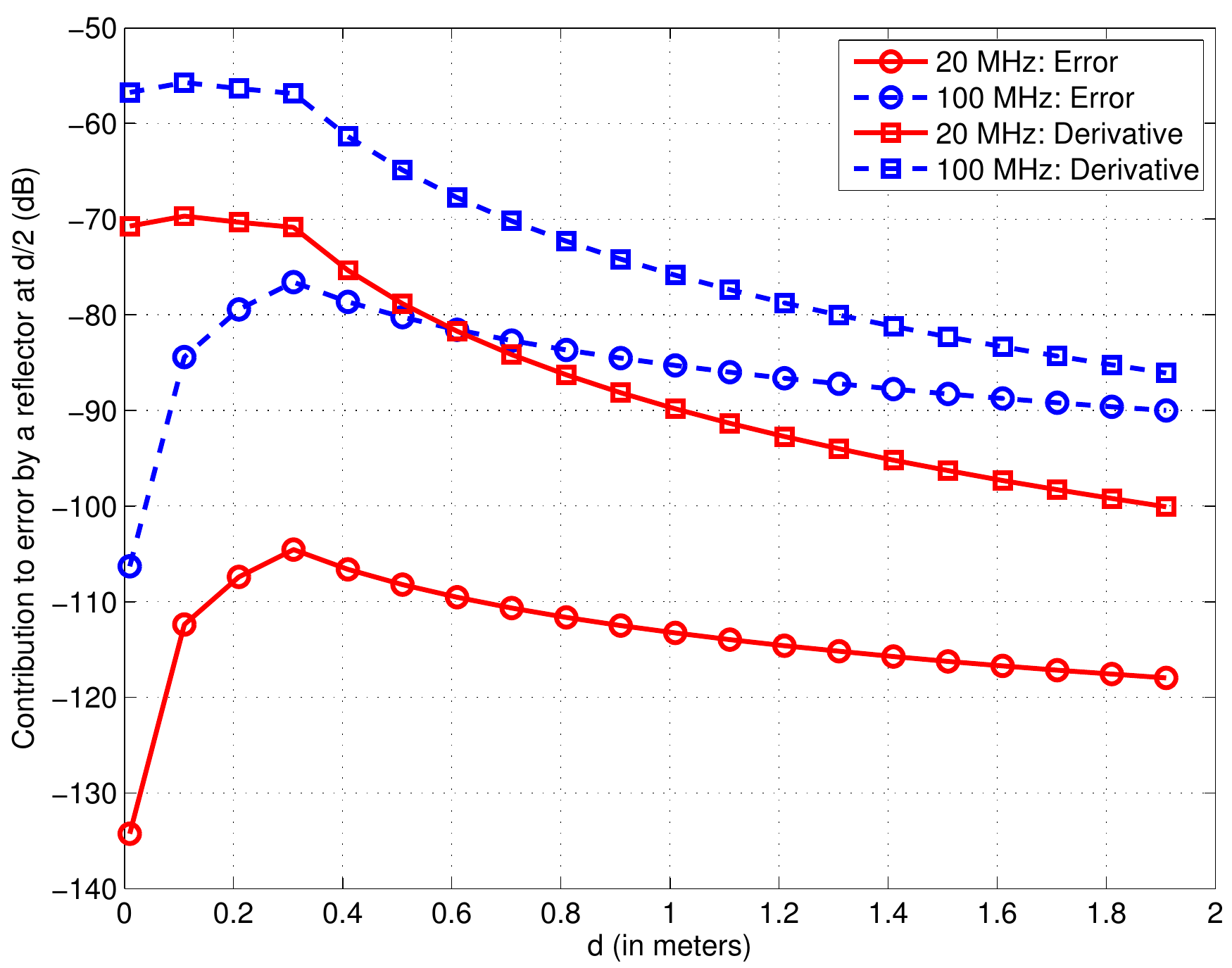}
\caption{Plot of $\l(d)^2 d^{4} ( cT)^{-4}$  with  measured and interpolated $\l(d)$ from Figure~\ref{fig:1}, assuming a $0$ dBm transmit signal.   }
\label{fig:error_dist}
\end{figure}
In  Figure \ref{fig:error_dist}, the contribution of a reflector at a distance $d/2$ towards the error term $|a_kE(\tau_k, t)|^2$ is plotted as a function of the distance $d$. A fitted model for  the measured path loss function in Figure~\ref{fig:1} is used for computing $a_k^2$. So for a transmit signal of 20 MHz bandwidth, we observe that the maximum contribution by any reflector or leakage to the error term is about  -100dB. So if the transmit signal has  a transmit power of 20 dBm, the error power would be about -80 dBm. Even if we have ten such reflections, the worst case error power is about -70 dBm. Observe that this error power is in the range of the typical  received signal power.  As in Lemma \ref{lem:der}, we can show that  the power in the derivative term $\E|a_k\tau_k x'(t)|^2  =\Theta(a_k ^2(  \tau_k/T)^2 )$.
In Figure \ref{fig:error_dist},  we also observe that the derivative term $I_d(t)$ is about  20-30 dB higher than the error term.

\subsection{Higher order approximation}
Similar to the previous subsection, a better approximation of the self-interference signal can be realised by using higher order terms of the Taylor series, \ie,
\begin{align}
I(t) = I_{s}(t) - I_{d}(t)+I_{d_2}(t) +\hdots (-1)^nI_{d_n}(t) +E_{n+1}(t).
\end{align}
Here 
\[I_{d_n}(t)=\re\left[C_n \frac{\d^n x(t)}{\d x^n} e^{j2\pi f_c t}\right],\]
is the self-interference corresponding to the $n^{\text{th}}$ derivative term. $C_n$ represents the accumulated phase and amplitude shifts of the $n^{\text{th}}$ derivative because of the multipath and equals
\[C_n = \sum_{k=1}^M \frac{a_k \tau_k^n}{n!} e^{-j 2\pi f_c \tau_k}.\]
While, higher order derivatives provide better approximation, the overall complexity of the cancellation process increases. Similar to the proof of Lemma \ref{lem:der}, it can easily be proved that \[\E[|E_{n+1}(\tau_k, t)|^2]  = \Theta\left( \sum_{k=0}^M a_k ^2(  \tau_k/T)^{2(n+1)} \right),\]
which shows that the error decreases with increasing $n$. We now look at techniques for cancelling $I_s(t)$ and $I_d(t)$ in the next Section.

\section{Self-interference cancellation architecture}
\label{sec:model}
A significant part of the self-interference signal should be cancelled before the LNA to prevent saturation and non-linearities in various elements of the receiver chain. 
Cancelling a part of the SI signal in the RF domain also helps in suppressing the in-band transmitter noise. 
  So the larger component $I_s(t)$ of the  self-interference term  has to be cancelled before the LNA. The derivative signal  can be cancelled either in the RF domain (before LNA), or in the analog baseband domain (before Analog to Digital Converter (ADC) and after mixer) or in the digital baseband domain (after ADC). These three options lead to three different architectures. However, as a proof of concept and for ease of implementation, we cancel the derivative term in the digital domain.   See Figures \ref{fig:digital_can} and \ref{fig:analog_can} for illustration of the cancellation architectures.


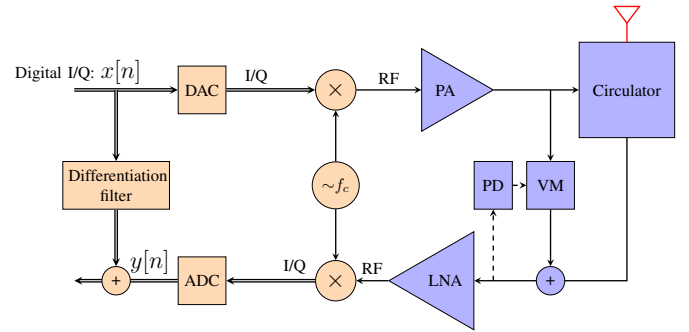
\begin{figure} [!htb]
\center
\resizebox{9cm}{!}{
\begin{tikzpicture}[node distance=0.8cm]
\node (null) {};
\node (DAC) [box1,   minimum width=1cm, right of=null, xshift=2cm] {DAC};
\draw [arrowdouble] (null) -- node[anchor=south, xshift =-1cm] {Digital I/Q: {\Large $x[n]$}}  ( DAC);
\node (mixerT) [circ1,   right of=DAC,   xshift=2cm] {$\mathlarger{\mathlarger{\mathlarger{\times}}}$};
\draw [arrowdouble](DAC)  -- node[anchor=south, xshift =-0.3cm] {I/Q}  (mixerT);
\node (PA) [triag1,   right of=mixerT,   xshift=1.5cm, rotate=30,minimum size=2cm] {\rotatebox{-30}{PA}};
\draw [arrow](mixerT)  -- node[anchor=south] {RF}  (PA);
\node (ANT) [ right of=PA,xshift=2cm]{};
\node(circulator) [box2,   minimum width=2cm, minimum height = 2cm ,right of=PA,  xshift=3cm] {Circulator};
\draw [arrow](PA)  -- node[anchor=south] {}  (circulator);
\draw [color=red,thick, right of=PA,xshift=2.1cm,yshift=0.1cm]   (8.8,0.9)\antenna;
\node [output, right of=PA, xshift=1.4cm] (output) {};
\node [output, left of=DAC, xshift=-1cm] (output2) {};

\node  (VM) [box2,   minimum width=1cm, right of=PA, xshift=1.4cm ,yshift=-2cm] {VM};
\node  (Diff) [box1,   minimum width=1cm, left of=DAC, xshift=-1cm ,yshift=-2cm,text width=2cm] {Differentiation  filter};
\draw [arrow] (output) -- node [name=y] {}(VM);
\draw [arrowdouble] (output2) -- node [name=y] {}(Diff);

\node (fc) [circ1,   right of=DAC, xshift=2cm,  yshift=-2cm] {$\sim$ \newline $ f_c$};
\draw [arrow] (fc) -- node [name=y] {}(mixerT);

\node (nullR)[yshift =-4cm] {};
\node (ADC) [box1,   minimum width=1cm, right of=nullR, xshift=2cm] {ADC};
\draw [arrowdouble] (ADC) -- node[anchor=south, xshift =0.5cm ] { \Large  $y[n]$}  (nullR);
\node (mixerR) [circ1,   right of=ADC,   xshift=2cm] {$\mathlarger{\mathlarger{\mathlarger{\times}}}$};
\draw [arrowdouble](mixerR)  -- node[anchor=south, xshift =.5cm] {I/Q}  (ADC);
\node (LNA) [triag1,   right of=mixerR,   xshift=1.5cm, rotate=-30] {\rotatebox{30}{LNA}};
 \draw [arrow](LNA)  -- node[anchor=south] {RF}  (mixerR);
  \draw [arrow](circulator)  |- node {}  (LNA);

    \node (sumR) [circ2,   right of=LNA,   xshift=1.4cm] {+};
        \node (sumR2) [circ1,   left of=ADC,   xshift=-1cm] {+};
              \node (PD) [box2,   minimum width=0.8cm, right of=VM, xshift=-2cm ,yshift=0cm] {PD};
 \draw [arrow,dashed] (8.9,-4) -- node[anchor=south] [name=y] {}(PD);
    \draw [arrow] (VM) -- node [name=y] {}(sumR);
     \draw [arrow,dashed] (PD) -- node [name=y] {}(VM);
        \draw [arrowdouble] (Diff) -- node [name=y] {}(sumR2);
        \draw [arrow] (fc) -- node [name=y] {}(mixerR);

\end{tikzpicture}
}
\caption{Cancellation architecture. We cancel $I_s(t)$ in the RF domain and $I_d(t)$ in the digital domain. Here VM represents vector modulator for phase shifting and attenuation, while PD represents power detector. The output of the power detector is proportional to the residual power after RF cancellation and is used to adjust the phase and attenuation of the vector modulator.}
\label{fig:digital_can}
\end{figure}

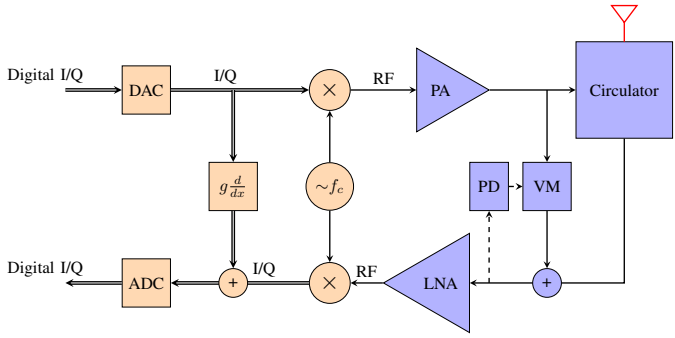
\begin{figure}[!htb]
\center
\resizebox{9cm}{!}{
\begin{tikzpicture}[node distance=0.8cm]
\node (null) {};
\node (DAC) [box1,   minimum width=1cm, right of=null, xshift=1cm] {DAC};
\draw [arrowdouble] (null) -- node[anchor=south, xshift =-1cm] {Digital I/Q}  ( DAC);
\node (mixerT) [circ1,   right of=DAC,   xshift=3cm] {$\mathlarger{\mathlarger{\mathlarger{\times}}}$};
\draw [arrowdouble](DAC)  -- node[anchor=south, xshift =-0.3cm] {I/Q}  (mixerT);
\node (PA) [triag1,   right of=mixerT,   xshift=1.5cm, rotate=30,minimum size=2cm] {\rotatebox{-30}{PA}};
\draw [arrow](mixerT)  -- node[anchor=south] {RF}  (PA);
\node (ANT) [ right of=PA,xshift=2cm]{};
\node(circulator) [box2,   minimum width=2cm, minimum height = 2cm ,right of=PA,  xshift=3cm] {Circulator};
\draw [arrow](PA)  -- node[anchor=south] {}  (circulator);
\draw [color=red,thick, right of=PA,xshift=2.1cm,yshift=0.1cm]   (8.8,0.9)\antenna;
\node [output, right of=PA, xshift=1.4cm] (output) {};
\node [output, right of=DAC, xshift=1cm] (output2) {};

\node  (VM) [box2,   minimum width=1cm, right of=PA, xshift=1.4cm ,yshift=-2cm] {VM};
\node  (Diff) [box1,   minimum width=1cm, right of=DAC, xshift=1cm ,yshift=-2cm] {$g\frac{d}{dx}$};
\draw [arrow] (output) -- node [name=y] {}(VM);
\draw [arrowdouble] (output2) -- node [name=y] {}(Diff);

\node (fc) [circ1,   right of=DAC, xshift=3cm,  yshift=-2cm] {$\sim$ \newline $ f_c$};
\draw [arrow] (fc) -- node [name=y] {}(mixerT);

\node (nullR)[yshift =-4cm] {};
\node (ADC) [box1,   minimum width=1cm, right of=nullR, xshift=1cm] {ADC};
\draw [arrowdouble] (ADC) -- node[anchor=south, xshift =-1cm ] {Digital I/Q}  (nullR);
\node (mixerR) [circ1,   right of=ADC,   xshift=3cm] {$\mathlarger{\mathlarger{\mathlarger{\times}}}$};
\draw [arrowdouble](mixerR)  -- node[anchor=south, xshift =.5cm] {I/Q}  (ADC);
\node (LNA) [triag1,   right of=mixerR,   xshift=1.5cm, rotate=-30] {\rotatebox{30}{LNA}};
 \draw [arrow](LNA)  -- node[anchor=south] {RF}  (mixerR);
  \draw [arrow](circulator)  |- node {}  (LNA);

    \node (sumR) [circ2,   right of=LNA,   xshift=1.4cm] {+};
    \node (PD) [box2,   minimum width=0.8cm, right of=VM, xshift=-2cm ,yshift=0cm] {PD};
    \draw [arrow,dashed] (8.9,-4) -- node[anchor=south] [name=y] {}(PD.south);
        \node (sumR2) [circ1,   right of=ADC,   xshift=1cm] {+};
    \draw [arrow] (VM) -- node [name=y] {}(sumR);
    \draw [arrow,dashed] (PD) -- node [name=y] {}(VM);
        \draw [arrowdouble] (Diff) -- node [name=y] {}(sumR2);
        \draw [arrow] (fc) -- node [name=y] {}(mixerR);
  \end{tikzpicture}}
\caption{Canceling $I_d(t)$ in the analog domain. The derivative canceler is implemented in   the analog domain after down conversion and before the  ADC.  This block is represented as $g\frac{d}{dx}$, where $g$ represents the tunable complex gain.}
\label{fig:analog_can}
\end{figure}

\subsection{Cancellation of the signal term $I_s(t)$}
The signal term $I_s(t)$ in the received self-interference can be rewritten as 
\[I_s(t)= \re[|C_0|  x(t)e^{j2\pi f_c t+j \arg(C_0)} ],\]
where $|C_0|$ represents the absolute value of the complex gain $C_0$ and $\arg$ denotes its angle (argument). So $I_s(t)$ can be obtained by scaling the known transmitted RF signal $y(t)=\sqrt{G_t}\re[x(t)e^{j2\pi f_ct }]$ by $|C_0|$ and phase-shifting the carrier by $\arg(C_0)$. 

The phase of an RF signal (carrier)  can be changed  by using a vector modulator (VM) \cite{razavi1998rf}.  See Figure \ref{fig:VM} for a block level illustration of a vector modulator. A vector modulator changes the phase of the RF carrier and the scales the amplitude of the signal. The output of a vector modulator  with $\re[x(t)e^{j2\pi f_c t}]$ as the input RF signal is  
\[Z(t)=\beta\re[ x(t)e^{j2\pi f_c t+j\theta }],\]
where the phase $\theta$ and the gain (or attenuation) $\beta$ can be  changed by changing the gains of the two arms in the vector modulator.  

To obtain the values of $G_1$ and $G_2$, we employ a \emph{training} period. In the training period, we ensure that there is no received signal at the device (\ie, no clients are sending anything to the device during that short period). Then the signal coming out of port 3 will contain only the self-interference signal $I_s(t)$. The values of $G_1$ and $G_2$ are then estimated based on the output of a power detector which measures the power in the residual self-interference. Once training is done, the values of $G_1$ and $G_2$ are adaptively changed based on the output of the power detector.
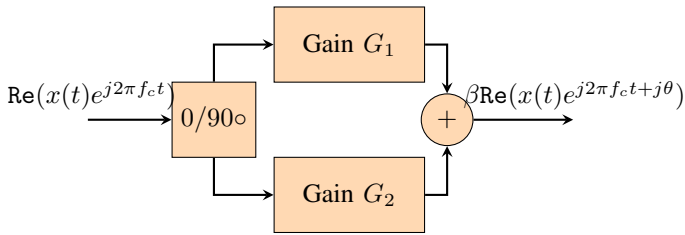
\begin{figure}[!ht]
\center
\begin{tikzpicture}[node distance=0.8cm]
\node (null) {};
\node (ninety) [box1,   minimum width=1cm, right of=null, xshift=1cm] {$0/90{\circ}$};
\node (gain1) [box1, minimum width=2cm, right of=ninety,  yshift=-1cm, xshift=1cm] {Gain $G_2$};
\node (gain2) [box1, minimum width=2cm,  right of=ninety, yshift=1cm, xshift=1cm] {Gain $G_1$};
\node (add) [circ1 ,  right of=gain1, yshift=1cm, xshift=0.5cm] {$+$};
\node (null2) [right of=add,xshift=1cm]{};
\draw [arrow] (null) -- node[anchor=south, xshift =-0.5cm] {$\re(x(t)e^{j2\pi f_c t})$}  (ninety);
\draw [arrow] (ninety) |- (gain1);
\draw [arrow] (ninety) |- (gain2);
\draw [arrow] (gain2) -| (add);
\draw [arrow] (gain1) -| (add);
\draw [arrow] (add);
\draw [arrow] (add) -- node[anchor=south, xshift=0.68cm] {$\beta\re(x(t)e^{j2\pi f_c t+j \theta})$} (null2);
\end{tikzpicture}
\caption{Illustration of a RF  vector modulator that can be used for phase shifting a signal.  The block $0/90^{\circ}$  splits the incoming signal into two RF signals whose RF phase differ by $90^{\circ}$. The output phase  $\theta$ and the gain $\beta$ depend on the gain $G_1$ and gain $G_2$.}
\label{fig:VM}
\end{figure}
The RF input to the power detector is 
\[r(t)=I_s(t) -I_d(t) + \beta\re[ x(t)e^{j2\pi f_c t+j\theta }].\]
Hence the output of the power detector, with an integrating time constant $\nu$ is 
\[P(t) = \frac{1}{\nu}\int_{t}^{t+\nu} |r(a)|^2 \d a.\]
Since, the power in an RF signal is twice that of the corresponding baseband, 
\begin{align*}
P(t) &= \frac{2}{\nu}\int_t^{t+\nu} |x(a)(C_0  + \beta e^{j\theta})-C_1x'(a)|^2 \d a
\end{align*}
\begin{align*}
&= |(C_0  + \beta e^{j\theta})|^2\frac{2}{\nu}\int_t^{t+\nu} |x(a)|^2\d a\\
&+  |C_1|^2\frac{2}{\nu}\int_t^{t+\nu} |x'(a)|^2\d a\\
& -C_1(C_0+\beta e^{j\theta})\frac{4}{\nu}\int_t^{t+\nu}\re[x(a)x'(a)]\d a.
\end{align*}
If  the integrating time constant $\nu$ is much larger than the symbol duration of the training signal $x(t)$, we have 
\begin{align*}
P(t) & \approx  2|(C_0  + \beta e^{j\theta})|^2E_s +2|C_1|^2E_d -4C_1(C_0+\beta e^{j\theta})g(t),
\end{align*}
where $E_s$ and $E_d$ are the powers of the signals $x(t)$  and $x'(t)$ and $g(t)$ is a high frequency signal with twice the frequency content of the signal $x(t)$.
Since $|C_1| \ll |C_0|$, we have
\begin{align*}
P(t) \approx  2|(C_0  + \beta e^{j\theta})|^2E_s.
\end{align*}
 Hence, for a large  integrating time constant, the output of the power detector is proportional to $ 2|(C_0  + \beta e^{j\theta})|^2$ and the minimum value is obtained for $\beta$ and $\theta$ such that $C_0  =-\beta e^{j\theta}$. So in the experimental setup (described later), the value of $\beta$ and $\theta$ are changed by changing the value of $G_1$ and $G_2$ based on the output of the power detector.

\subsection{Cancellation of the derivative term $I_d(t)$}

\subsubsection{Analog domain cancellation}    The derivative term can be cancelled in the analog domain (after downconversion) or in  the digital domain.  However, cancelling more self-interference in the analog domain before the ADC would increase the bit resolution of the digitized received signal.  A differentiator can be easily implemented in the analog domain using a resistor in series with a capacitor across a large gain  amplifier. See Figure \ref{fig:dx}.  The unity gain frequency of the op-amp and the values of $R$ and $C$  determine the bandwidth and the gain of the differentiator.  In  practice, a resistor in series with the capacitor C and a capacitor in parallel with the resistor R are added to reduce gain for high frequency noise. 
 \begin{figure}[ht]
 \centering
\includegraphics[width=0.5\columnwidth]{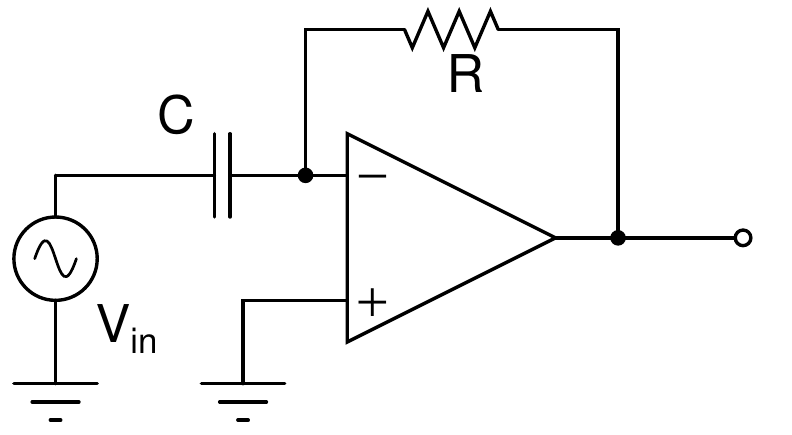}
 \caption{Schematic of an analog differentiator.}
 \label{fig:dx}
\end{figure}
 The real and imaginary parts  (I and Q) of the transmitted baseband  signal have to be scaled jointly so as to  obtain a replica of the derivative term in the self-interference signal.  
 There are significant  advantages of cancelling the derivative term in the analog domain:
\begin{enumerate}
\item Canceling more self-interference before the analog to digital converter would increase the bit resolution of the received signal, thereby improving the effective received signal-to-noise ratio (SNR).
\item Digital cancellation would require realizing the filter $j \omega$, which cannot be realized when the sampling rate is equal to the bandwidth. Realizing an approximate filter would require oversampling the signal thereby increasing the complexity of the ADC.
\end{enumerate}

As shown in Figure \ref{fig:analog_can}, the derivative circuit should be between the ADC, DAC and the mixers. However, such placement is very difficult in off-the shelf equipment (like USRP) and hence, in this paper,  for simplicity of implementation we restrict ourselves to the cancellation of the derivative term in the digital domain.

\subsubsection{Digital domain cancellation}


From \eqref{eqn:Eqn_SelfInterf}, the self-interference signal in the baseband is given by $I(t)= c_0 x(t) -c_1x'(t)+ e_b(t)$.
Theoretically, the RF/analog cancellation should have removed the component $c_0x(t)$. However, in practice, because of the  gain and phase quantization  of the vector modulator and device imperfections, there will be a residual signal component even after RF cancellation. Hence the self-interference  before the ADC is  $r(t) \approx a_0 x(t) -c_1 x'(t)$, 
where for notational simplicity, we have neglected the error term $e_b(t)$. 

A digital domain differentiator can be realised by any filter with response $j\omega$ in the frequency domain. As mentioned previously, this filter cannot be realised if the sampling  rate is equal to the Nyquist rate of the signal. But a good approximation of the derivative can be obtained if the signal is oversampled\footnote{Most receivers oversample the signal for  timing and frequency synchronization.}. 
 A simple three tap digital domain filter that mimics a derivative  is 
\begin{align}
H_{d} =[-1, 0, 1].
\label{eq:der1}
\end{align}
 A better noise-robust nine tap approximation of the derivative filter \cite{snrd} is 
\begin{align}
H_d =[3, -32, 168, -672, 0, 672, -168, 32, -3]/840.
\label{eq:der2}
\end{align}
The frequency response of these filters are plotted in  Figure \ref{fig:Fig10_Der} and it can be  observed that an oversampling factor of $4$ would suffice for both these filters.
\begin{figure}[ht!]
\centering
 \includegraphics[width = 0.95\columnwidth]{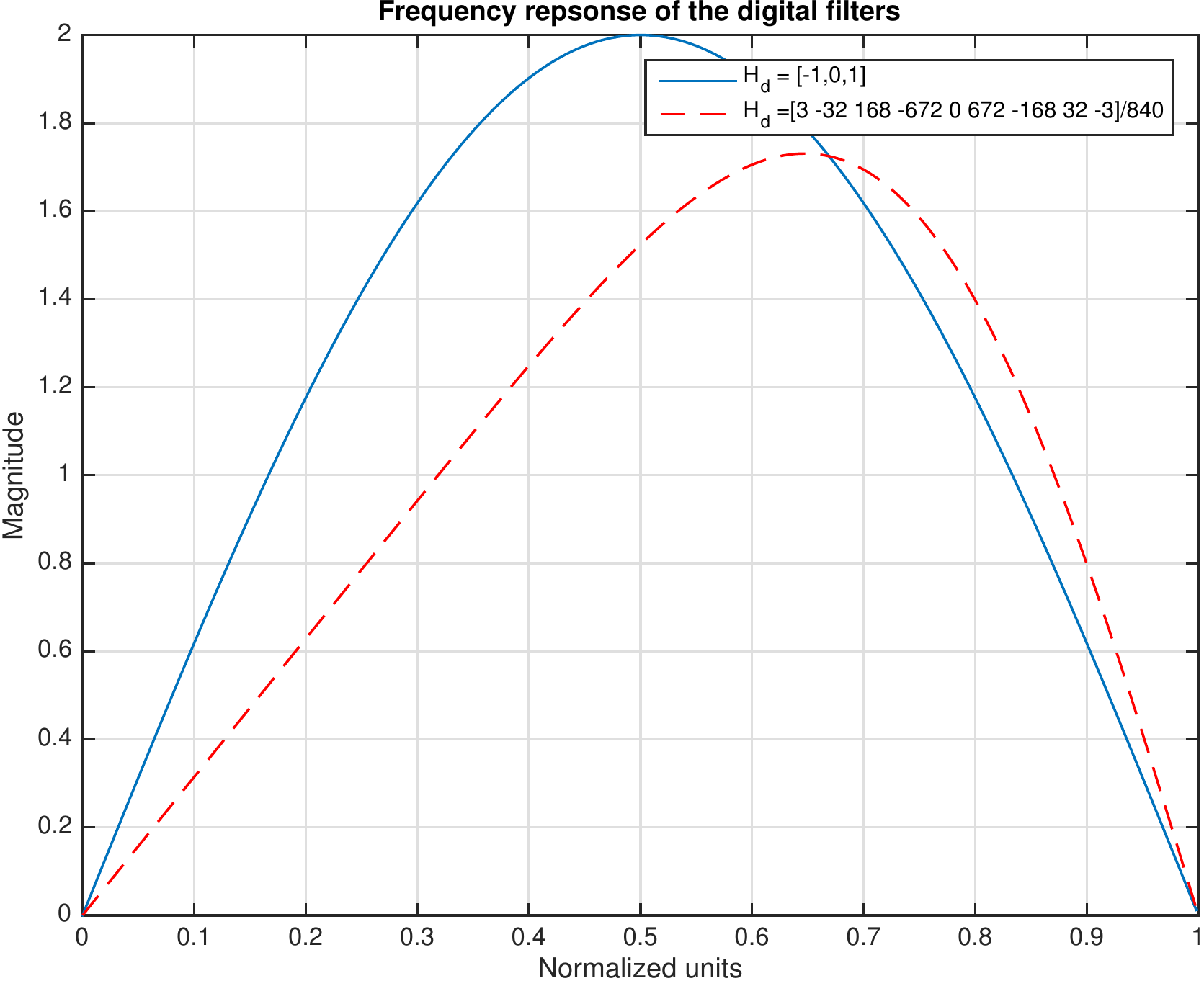}
 \caption{Frequency response of the derivative filters in \eqref{eq:der1} and \eqref{eq:der2}. We observe a good linear approximation till the normalized frequency of $0.3$. }
 \label{fig:Fig10_Der}
\end{figure}
The derivative of the transmitted signal in the digital domain is given by $x'[n]= x[n]\conv H_d[n]$.
 Let $y[i]$ denote the received complex samples in the digital domain. See Figure \ref{fig:digital_can}. In the training phase, the coefficients $a_0$ and $c_1$ are chosen so as to minimize the mean squared error 
$\sum_{i=1}^N| y[i]- a_0x[i]-c_1x'[i]|^2$.
Let $X = [x[1], \hdots, x[N]]^T$,  $X_1= [x'[1], \hdots, x'[N]]^T$ and let $Y= [y[1], \hdots, y[N]]^T$. Then least squares (LS) estimates of $a_0$ and $c_1$ are given by the solutions of 
\begin{align} \underbrace{\left[ 
\begin{array}{cc}
X^HX&X_1^HX\\
X^HX_1&X_1^HX_1
\end{array}
\right]}_{\mathcal{X}}
\left[\begin{array}{c}
\hat{a_0}\\
\hat{c_1}
\end{array}
\right]=
\underbrace{\left[\begin{array}{c}
X^HY\\
X_1^HY
\end{array}
\right]}_{\mathcal{Y}},
\label{eq:mat}
\end{align}
where $X^H$ represents the conjugate transpose of the vector $X$. Note that $\mathcal{X}$ will be a $2\times2$ matrix and $\mathcal{Y}$ will be $2\times1$ vector.
Using these estimates of $\hat{a_0}$ and $\hat{c_1}$, the reconstructed self-interference signal after the training phase is \[\hat{I}[n]=\hat{a_0}x[n] -\hat{c_1}x'[n], n= N+1, N+2, \hdots .\]  $\hat{I}[n]$ is then subtracted from the received signal $y[n]$ (after the training phase) to cancel the self-interference signal.

{\em Complexity:} Observe that the  inverse of the matrix $\mathcal{X}$ can be precomputed and stored. Only the matrix $\mathcal{Y}$ has to be computed based on the received signal. Computing each term of the matrix $\mathcal{Y}$ in \eqref{eq:mat} requires approximately $N$ multiplications and  $N$ additions and obtaining the coefficients would require a $2\times 2$ matrix multiplication with a $2 \times 1$ vector. Hence the computational complexity (complex operations) of the procedure scales as $4N+8$ irrespective of the number of multipaths $M$ in the self-interference channel. The complexity of computing the derivative is $2LN$ with a filter length $L$. Hence the total complexity of the proposed digital cancellation is $(2L+4)N+8$ complex operations.
On the other hand  channel estimation, without any prior model on the channel taps,  assuming a filter length $K$ requires about $2KN+2K^2$ complex computations. In earlier implementations, typically more than $30$ taps are assumed, \ie,  $K\geq 30$.

%
%

We now look at the case, when the second derivative is used in-addition to the first derivative to approximate the delayed signal. In this case, the self-interference signal before the ADC is $I(t)= a_0 x(t) -c_1x'(t)+ c_2x''(t)+e_{2D}(t)$.
The second derivative in the digital domain can be approximated by passing the signal through the filter \cite{snrd} 
\[H_{d_2}=[1, 4, 4, -4, 10, -4, 4, 4, 1 ]/64.\]  
Let  $x''[n]= x[n]\conv H_{d_2}[n]$, 
 and $X_2 = [x''[1], \hdots, x''[N]]^T$. Then the LS estimate of the coefficients are obtained as the solution of
\[ \left[ 
\begin{array}{ccc}
X^HX&X_1^HX&X_2^HX\\
X^HX_1&X_1^HX_1&X_2^HX_1\\
X^HX_2&X_1^HX_2&X_2^HX_2
\end{array}
\right]
\left[\begin{array}{c}
\hat{a_0}\\
\hat{c_1}\\
\hat{c_2}
\end{array}
\right]=
\left[\begin{array}{c}
X^HY\\
X_1^HY \\
X_2^HY
\end{array}
\right].
\]
Once $\hat{a_0}$, $\hat{c_1}$ and $\hat{c_2}$ have been obtained, the self-interference signal can be reconstructed and subtracted from the received signal.

 \section{Experimental results}
 \label{sec:exp}
 In this Section, we provide experimental results to validate the channel model and the derivative based cancellation.  We demonstrate the ability of the proposed derivative based architecture to suppress the self-interference signal and we provide quantitative measure in terms of cancellation for the same.  
  \subsection{Experimental setup}
  Our experimental setup is shown in Fig.~\ref{fig:expt_setup}.  We use  a  National Instruments (NI) PXIe based software defined radio (NI5791) for transmission and reception. 
The maximum transmit power possible in NI5791 is 5 dBm and we use an external power amplifier (PA) (Skyworks SE2576L) at the transmitter. We use a shared antenna architecture wherein the same antenna is used for transmission and reception. The isolation between transmit and receive chain is provided by a circulator (Pasternack PE8401). This circulator provides 18 dB  of isolation between port 1 and port 3. The transmit signal is fed into port 1 and the antenna is connected to port 2. The signal from port 3 will therefore contain the received signal as well as the self-interference signal. Two copies of the transmit signal from PA are obtained using a directional coupler (Mini-Circuits ZHDC-16-63-S+), wherein one is connected to the input port of the circulator and the other is used as an input to the vector modulator (Hittite HMC631LP3).  The directional coupler allows for tapping a copy of a signal with minimal loss in the mainline, thus the output at the coupled port is much lower in power\footnote{A lower power at the input of the VM is desirable since the P1 dB of the VM that we use is 21 dBm and a lower power at its input prevents any significant non-linearity at the output of the VM.}. The VM also introduces an attenuation and thus the power at its output  may be insufficient to suppress the self-interference. The output of VM is passed through an amplifier (Mini-Circuits ZX60-P33ULN+) in order to recover this loss in power. 
   The self-interference signal at port 3 of the circulator is comprised of the transmit signal leaked through the circulator and multiple reflected copies of the transmit signal received by the antenna. 
\begin{figure}[!htb]
\centering
\includegraphics[width=3.5in]{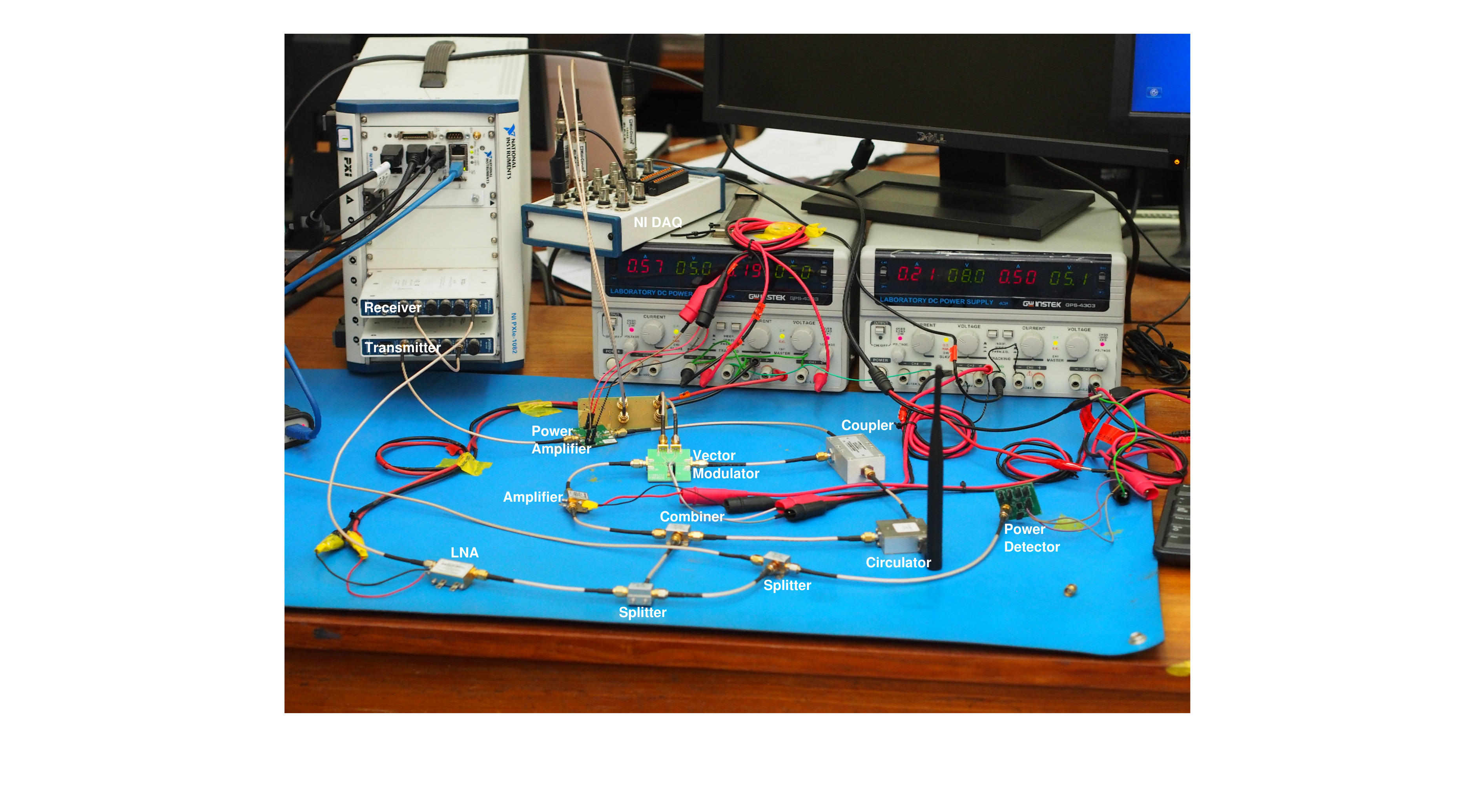}
\caption{Experiment Setup}
\label{fig:expt_setup}
\end{figure}
 
 As mentioned earlier, the  signal  component ($I_s(t)$) of the self-interference is cancelled in  the RF domain.  The vector modulator  is used to  adapt the gain and phase of the tapped transmitted signal and match it to the  $I_s(t)$ component of the self-interference. The gain and phase of the VM are controlled by two DC voltages  generated by an NI Data Acquisition Device (DAQ)\footnote{ Consists of 16-bit analog-to-digital converters (ADC) and 16-bit digital-to-analog converters (DAC) controllable by a desktop computer}. 
  The output of the VM and the self-interference signal (from the receive port of the circulator) are summed by a power combiner (Mini-Circuits ZX10-2-232-S+). A part of this summed signal is fed to a true RMS power detector (PD) (Hittite HMC1020LP4ETR) via a power splitter to observe the power in the residual signal. The PD generates a DC voltage proportional to the input power. This voltage is sampled by the NI DAQ. The optimal DC control voltages of the VM are found by an adaptive search that minimizes the residual self-interference power. 

The residual signal after the combiner is fed to the NI5791 receiver\footnote{The minimum RF power at the input of 5791, that induces full-scale swing at the ADC is -27 dBm. Since the residual self-interference signal is much lower in power, we use an (Mini-Circuits ZX60-242GLN-S+) before the NI5791 module, to prevent effect of quantization noise.}.
 The received samples comprise of the signal term and the derivative term. A part of these samples are training symbols. They are processed offline to obtain $\hat{a_0}$, $\hat{c_1}$ and $\hat{c_2}$. These estimated parameters are then used to reconstruct and cancel self-interference for the remaining samples. ((Non-linear cancellation of the third and fifth harmonic is also used to mitigate the non-linear effects of the PA.))
 
 We obtain the cancellation results for  OFDM and single-carrier modulated waveforms.
\subsubsection{OFDM}
We consider an OFDM signal with 1024 subcarriers of which  620 are useful subcarriers (the rest are nulled out at the DC and at the edge of the band). 
At the receiver we use an oversampling factor of 4. The maximum sampling rate that can be practically achieved using PXIe is 80 MS/s. Hence with an oversampling factor of 4, the maximum bandwidth of OFDM signal that can be transmitted is 20 MHz. The P1dB of the PA is 32 dBm.  The measured PAPR (Peak to Average Power Ratio) of the  transmitted OFDM  waveform was 13 dB and hence the maximum average transmit power was restricted to 19 dBm to avoid severe non-linearities.
The spectrum of  a 20 MHz OFDM signal that is used is plotted in  Figure \ref{fig:Fig_spectrum_deriv_ofdm}.
 
\subsubsection{Single-carrier}  
A 4-QAM single-carrier signal was also used for the experiments.  An RRC pulse shaping filter with roll-off factor 0.3 was used. The PAPR of the signal was measured to be 4 dB which is about 9 dB lower than that of OFDM.  Hence, with the same PA, the single carrier can be transmitted at higher power than OFDM without PA saturation. 
 
We use 2.395 GHz as the center frequency for all the experiments.  This was done mainly to avoid interference from the ISM band.

 \subsection{Results and discussion}
%
Following the standard convention in literature, we define cancellation to be ratio of power of the SI signal after cancellation to the power of the SI signal before cancellation. The ratio is expressed in dB. Note that the power of the SI signal before cancellation is the same as transmit power at the antenna port (port 2) of the circulator. 

 In Figure  \ref{fig:Fig_spectrum_deriv_ofdm}, the spectrum of the self-interference signal is plotted when an OFDM signal is transmitted at 4 dBm (at port 2 of the circulator). In the same figure, the residual self-interference is plotted after analog cancellation. About 54 dB of self-interference was cancelled in the analog domain.   More importantly, the linear slope in the residual self-interference spectrum indicates that the residual signal is dominated by the derivative component $I_d(t)$, thus verifying the derivative approximation and in particular \eqref{eqn:Eqn_SelfInterf}.
\begin{figure}[!htb]
\centering
\includegraphics[width=0.95\columnwidth]{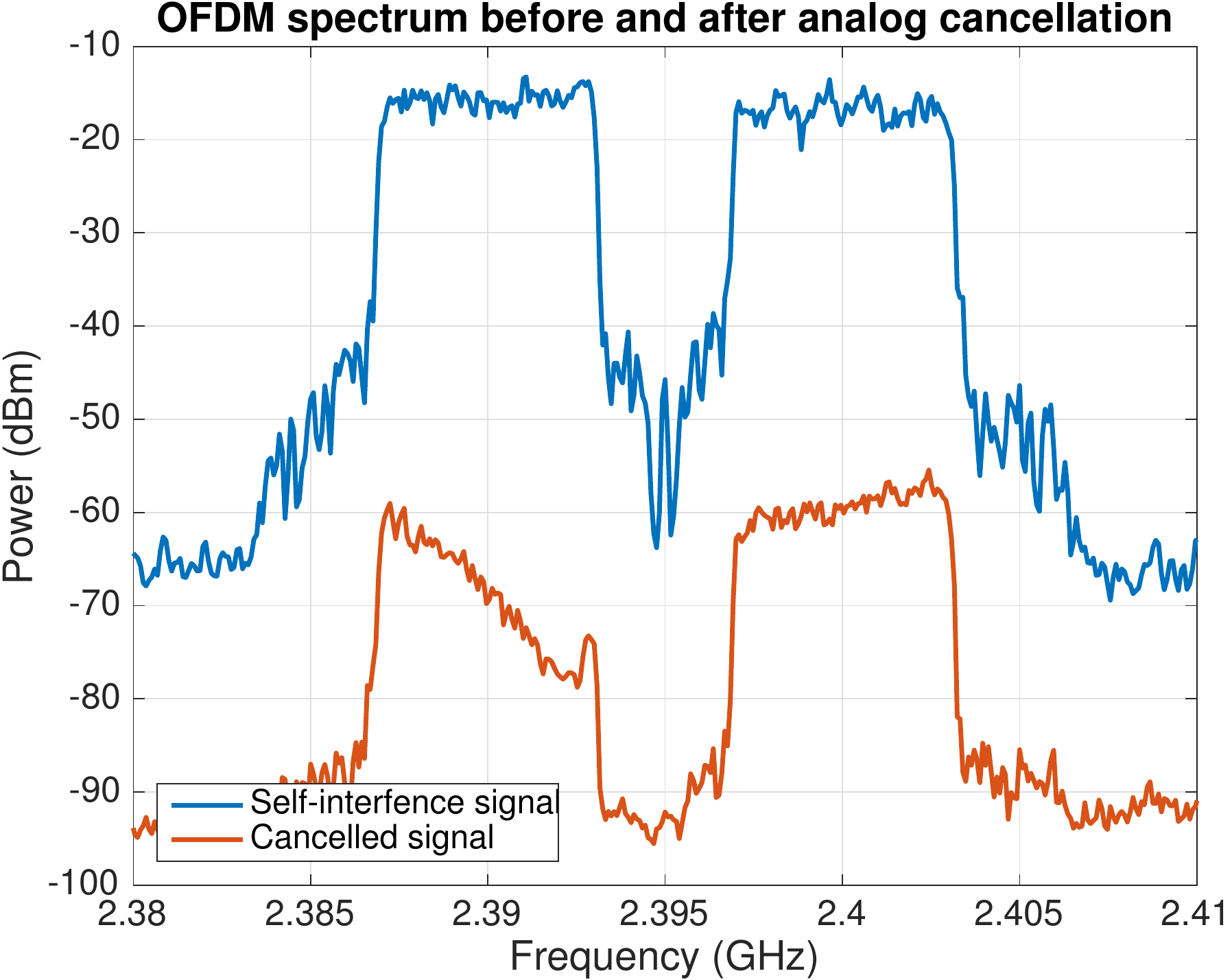}
\caption{ The spectrum of the self-interference signal of a 20 MHz OFDM transmission. Also, the  self-interference spectrum  after analog cancellation is  plotted. The transmit power (at the antenna port) is 4 dBm and the analog cancellation is about 54 dB. The linear slope in the residual self-interference indicates a derivative component. }
\label{fig:Fig_spectrum_deriv_ofdm}
\end{figure}
\begin{figure}[!htb]
\centering
\includegraphics[width=0.95\columnwidth]{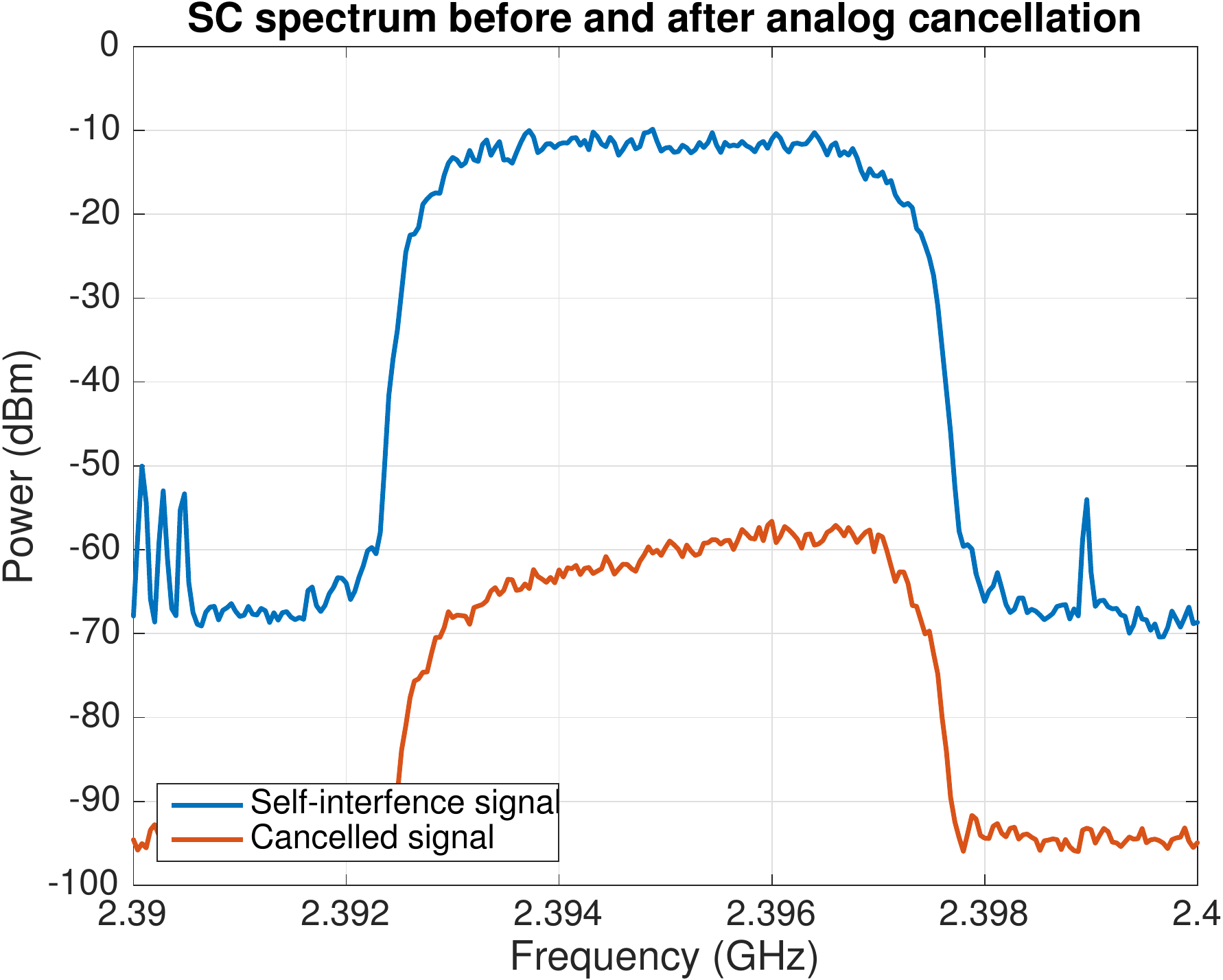}
\caption{The spectrum of the self-interference and the residual signal of a 10 MHz 4-QAM single carrier transmission.  The transmit power (at the antenna port) is 4 dBm and the analog cancellation is about 57 dB. The linear slope in the residual self-interference indicates a derivative component.}
\label{fig:Fig_spectrum_deriv_sc}
\end{figure}
In  Figure \ref{fig:Fig_spectrum_deriv_sc}, the spectra of the  self-interference and the cancelled signal (57 dB analog cancellation) are plotted when a single-carrier signal is transmitted. As in the OFDM signal, the residual self-interference exhibits a large derivative component. 
\begin{figure}[!htb]
\centering
\includegraphics[width = 0.95 \columnwidth]{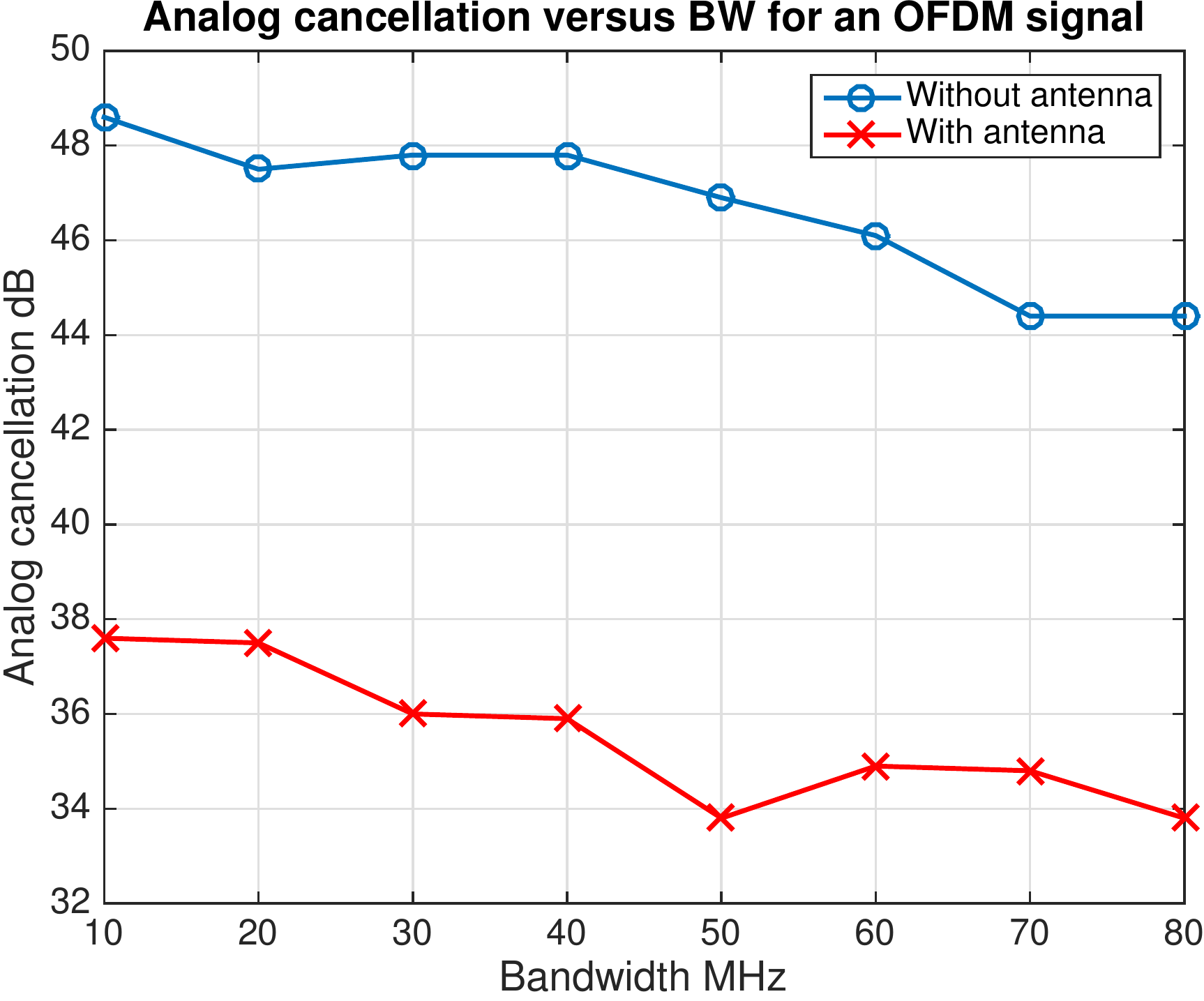}
\caption{Analog cancellation versus  transmit BW for an OFDM signal with and without antenna.  In the second case (without antenna), the antenna port is terminated by a 50 $\Omega$ terminator. }
\label{fig:Fig_cncln_BW}
\end{figure}

In Figure \ref{fig:Fig_cncln_BW}, the analog cancellation\footnote{The reported analog cancellation also includes the 18 dB isolation of the circulator.} is plotted as a function of the signal bandwidth. We observe that the analog cancellation decreases with increasing bandwidth. This is because the derivative component in the residual self-interference signal increases with increasing bandwidth. Since analog cancellation only removes $I_s(t)$, the residual power increases with increasing bandwidth, and thus lowering the  analog cancellation. The top curve in the plot corresponds to the case when the antenna port was terminated by a 50 $\Omega$ terminator, while the bottom curve corresponds to measurements with an antenna. In the case of 50 $\Omega$ termination, the self-interference multi-path  is primarily through the circulator, while with antenna, there will be multiple paths due to reflections too. In addition, the characteristic impedance of the antenna will not be as close to 50 $\Omega$ as a terminator. This impedance mismatch causes RF signals to get reflected back from the antenna (instead of getting transmitted). Because of these two effects, the aggregate  power in the derivative component increases in the case of antenna. This reduces causes the reduction in analog cancellation.

\begin{figure}[!htb]
\centering
\includegraphics[width = 0.95 \columnwidth]{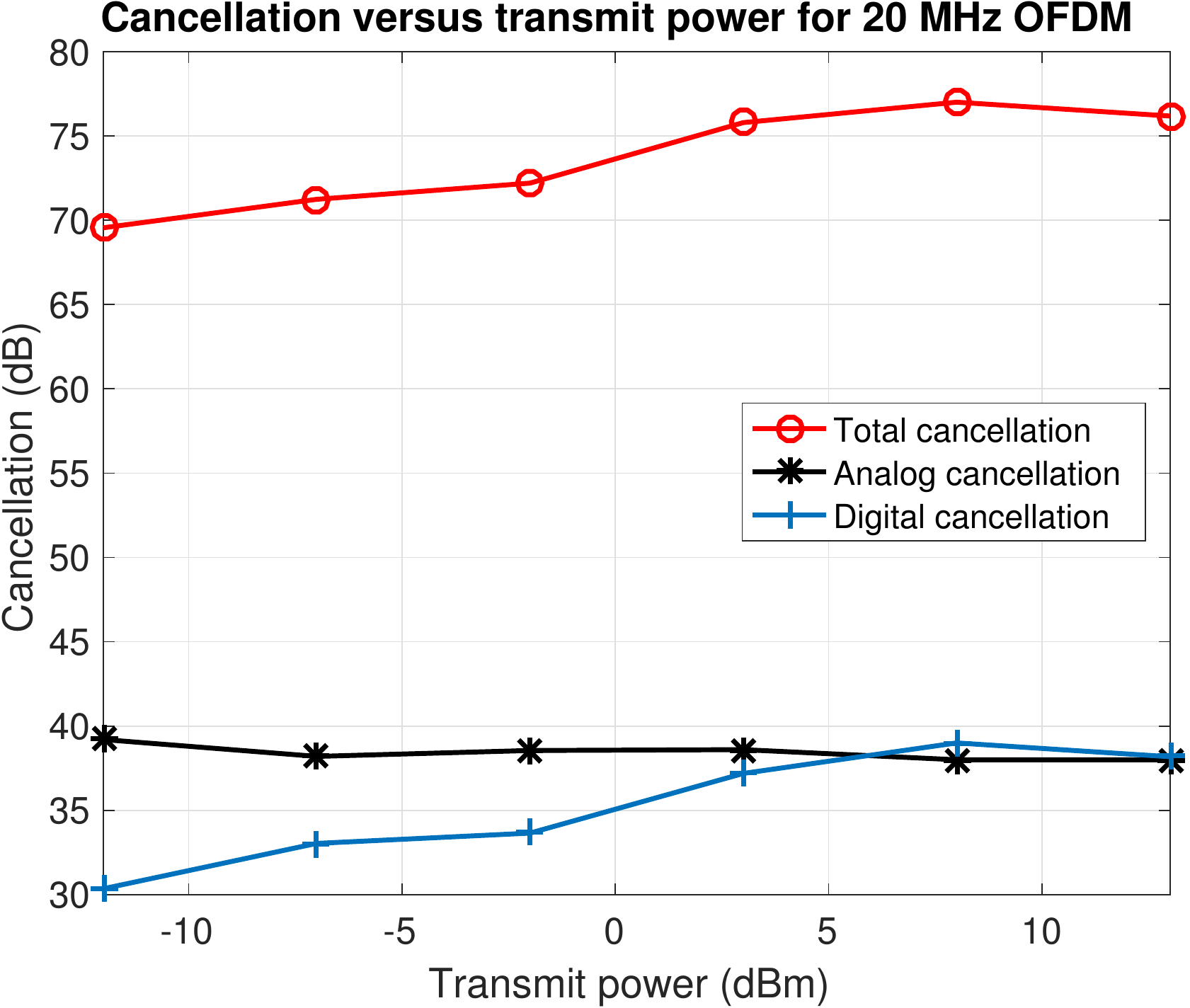}
\caption{Analog and digital cancellation versus transmit power for a 20 MHz OFDM signal with antenna.}
\label{fig:Fig_cncln_ofdm_term_total}
\end{figure}

In Figure \ref{fig:Fig_cncln_ofdm_term_total}, the analog and digital cancellation are plotted as a function of the transmit power. We observe that the analog cancellation is almost constant with respect to increasing transmit power. This is expected since the analog cancellation does not depend on the signal SNR and depends only on the resolution of the phase and amplitude of the VM, which are fixed. The digital cancellation  is increasing with the transmit power till about 5 dBm input power after which it reduces. This is mainly because of the power amplifier non-linearities. 
 \begin{figure}[!htb]
\centering
\includegraphics[width = 0.95 \columnwidth]{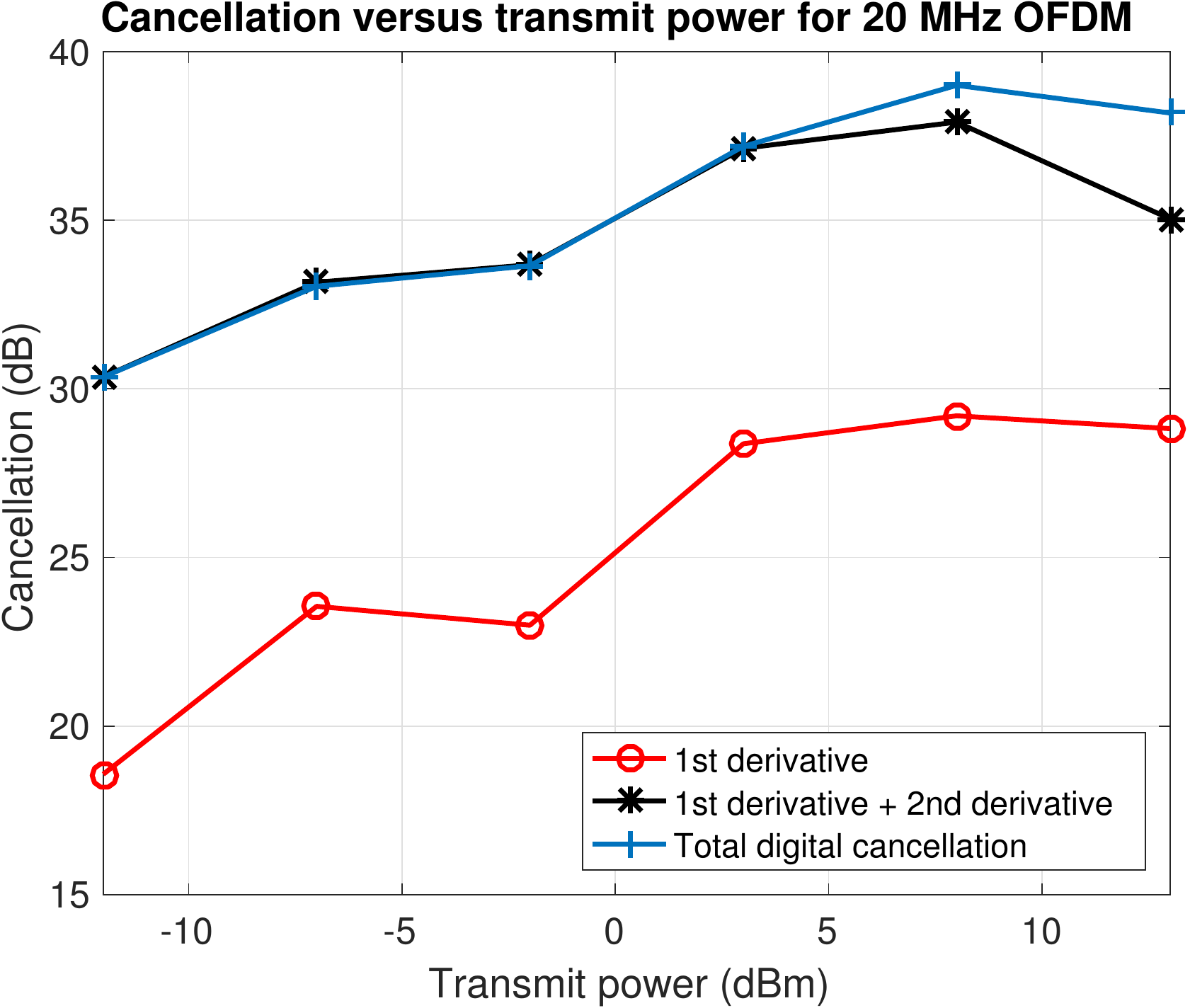}
\caption{Split-up of the digital cancellation versus transmit power for a 20 MHz OFDM signal with antenna.}
\label{fig:Fig_cncln_ofdm_term_split}
\end{figure}
As mentioned earlier, the digital cancellation consists of removing the signal, the derivative and the second order derivative components.  In Figure \ref{fig:Fig_cncln_ofdm_term_split}, this split-up is provided as a function of the transmit power. We see that the first-derivative cancellation provides the maximum cancellation. However, the second-derivative also gives about 5-6 dB of cancellation. The self-interference after analog cancellation is  $I(t)=I_s(t) + I_d(t)$, \ie, a sum of the the signal and the derivative terms. This signal is received in the digital domain after sampling by the ADC. Since the initial phase of the sampling time cannot be controlled, the received self-interference  in the digital domain  is $I(nT+\delta) =I_s(nT+\delta)+ I_d(nT+\delta)$, $n=1, 2, \hdots$, where $T$ is the sampling duration and $0\leq \delta \leq T$.  However we only have access to the transmitted signal $x(nT)$. Since $\delta$ is small, $I_s(nT+\delta)= a_0x(nT+\delta)$ can be approximated (after appropriate scaling by $a_0$) by $x(nT)$ and $x'(nT)$. Similarly $I_d(nT+\delta) = c_1x'(nT+\delta)$ can be approximated by $x'(nT)$ and the second derivative $x''(nT)$.  Hence using the second derivative improves the overall cancellation. 

\begin{figure}[!htb]
\centering
\includegraphics[width = 0.95 \columnwidth]{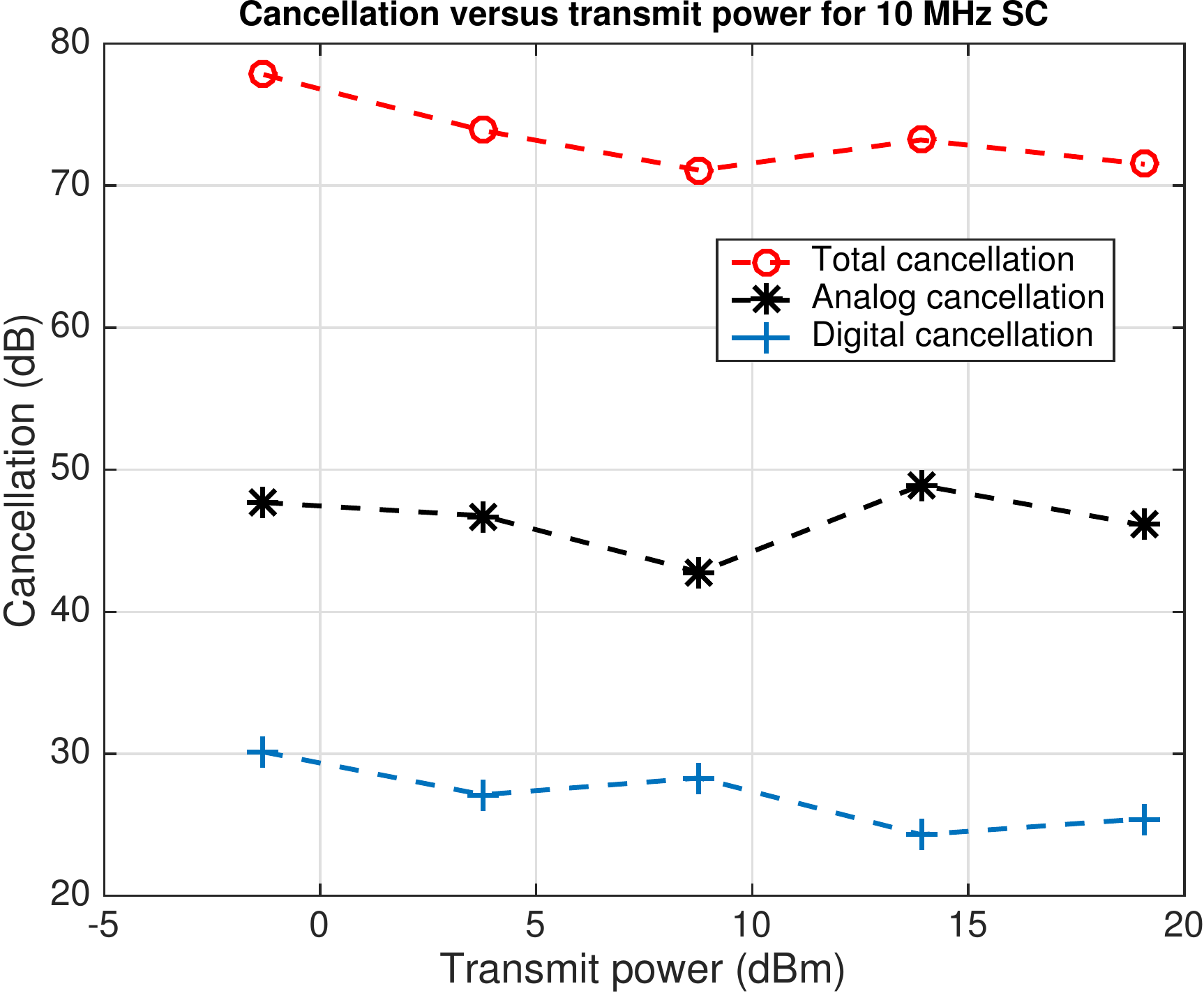}
\caption{Cancellation vs Transmit power for different  bandwidths with single carrier as transmit waveform. These results were obtained with port 2 of the circulator connected to an antenna.}
\label{fig:Fig_cncln_sc_ant}
\end{figure}

For a single carrier transmission, the cancellation is plotted as a function of the transmitted power in Fig.~\ref{fig:Fig_cncln_sc_ant}. Overall, we observe about 75 dBm cancellation for both OFDM and single-carrier waveforms. Thus the overall cancellation is impervious to the transmitted waveform. 

\section{Conclusion}
\label{sec:con}
Robust self-interference cancellation is critical to realising full-duplex capable wireless nodes. However a major impediment to  self-interference cancellation is estimation of the multi-path channel through which the transmitted signal reaches the  shared antenna. A part of the channel has to be estimated in the analog (RF) domain, and used for self-interference cancellation before the LNA, while the remaining part of the channel  has to be estimated in either the analog baseband or in the digital domain.  In the current literature, the channel is modelled as an $M$-tap delay-line filter and the filter coefficients are estimated in the RF and digital domain. In all the works, there was no prior  knowledge on $M$, and  hence a large number of taps are assumed.

In this work, using Talyor series approximation, we reduce the dimensionality of the parameter space to two (or three). In particular, we show that the self-interference can be modelled as a linear combination of the original signal and its derivatives. We propose a new self-interference cancellation architecture that utilises the linearized channel model. The self-interference model, and in particular the presence of the derivative component of the signal  is verified by experiments. About 75 dB of cancellation (without non-linearity cancellation) is observed for a 20 MHz wideband signal. With the linearization technique for canceling SI multipaths, much smaller form-factors could be realized, extending the use of FD capability to not just base-station or access points, but also to end user devices.
 
\section{Acknowledgement}
 The partial support of the Department of Electronics and Information Technology (DeitY), India through the 5G project is gratefully acknowledged. 

\bibliographystyle{IEEEtran}
\bibliography{FD}

\end{document}